\documentclass[11pt, draftcls, onecolumn]{IEEEtran}

\usepackage[mathscr]{eucal}
\usepackage[cmex10]{amsmath}
\usepackage{epsfig,epsf,psfrag}
\usepackage{amssymb,amsmath,amsthm,amsfonts,latexsym}
\usepackage{amsmath,graphicx,bm,xcolor,url}
\usepackage{fixltx2e}
\usepackage{array}
\usepackage{verbatim}
\usepackage{bm}
\usepackage{algorithmic, cite}
\usepackage{algorithm}
\usepackage{verbatim}
\usepackage{textcomp}
\usepackage{mathrsfs}
\usepackage{epstopdf}

\catcode`~=11 \def\UrlSpecials{\do\~{\kern -.15em\lower .7ex\hbox{~}\kern .04em}} \catcode`~=13 

\allowdisplaybreaks[4]

\newcommand{\calC}{\mathcal{C}}

\newcommand{\calE}{\mathcal{E}}

\newcommand{\calK}{\mathcal{K}}

\newcommand{\calM}{\mathcal{M}}
\newcommand{\calN}{\mathcal{N}}

\newcommand{\calS}{\mathcal{S}}
\newcommand{\calT}{\mathcal{T}}
\newcommand{\calU}{\mathcal{U}}
\newcommand{\calV}{\mathcal{V}}

\newcommand{\calX}{\mathcal{X}}
\newcommand{\calY}{\mathcal{Y}}
\newcommand{\calZ}{\mathcal{Z}}

\DeclareMathAlphabet{\mathbsf}{OT1}{cmss}{bx}{n}
\DeclareMathAlphabet{\mathssf}{OT1}{cmss}{m}{sl}

\DeclareSymbolFont{bsfletters}{OT1}{cmss}{bx}{n}  
\DeclareSymbolFont{ssfletters}{OT1}{cmss}{m}{n}
\DeclareMathSymbol{\bsfGamma}{0}{bsfletters}{'000}
\DeclareMathSymbol{\ssfGamma}{0}{ssfletters}{'000}
\DeclareMathSymbol{\bsfDelta}{0}{bsfletters}{'001}
\DeclareMathSymbol{\ssfDelta}{0}{ssfletters}{'001}
\DeclareMathSymbol{\bsfTheta}{0}{bsfletters}{'002}
\DeclareMathSymbol{\ssfTheta}{0}{ssfletters}{'002}
\DeclareMathSymbol{\bsfLambda}{0}{bsfletters}{'003}
\DeclareMathSymbol{\ssfLambda}{0}{ssfletters}{'003}
\DeclareMathSymbol{\bsfXi}{0}{bsfletters}{'004}
\DeclareMathSymbol{\ssfXi}{0}{ssfletters}{'004}
\DeclareMathSymbol{\bsfPi}{0}{bsfletters}{'005}
\DeclareMathSymbol{\ssfPi}{0}{ssfletters}{'005}
\DeclareMathSymbol{\bsfSigma}{0}{bsfletters}{'006}
\DeclareMathSymbol{\ssfSigma}{0}{ssfletters}{'006}
\DeclareMathSymbol{\bsfUpsilon}{0}{bsfletters}{'007}
\DeclareMathSymbol{\ssfUpsilon}{0}{ssfletters}{'007}
\DeclareMathSymbol{\bsfPhi}{0}{bsfletters}{'010}
\DeclareMathSymbol{\ssfPhi}{0}{ssfletters}{'010}
\DeclareMathSymbol{\bsfPsi}{0}{bsfletters}{'011}
\DeclareMathSymbol{\ssfPsi}{0}{ssfletters}{'011}
\DeclareMathSymbol{\bsfOmega}{0}{bsfletters}{'012}
\DeclareMathSymbol{\ssfOmega}{0}{ssfletters}{'012}

\DeclareMathOperator*{\argmax}{arg\,max}

\DeclareMathOperator{\var}{\mathsf{Var}}

\newtheorem{theorem}{Theorem} 
\newtheorem{lemma}[theorem]{Lemma}

\newtheorem{remark}{Remark}

\newcommand{\qednew}{\nobreak \ifvmode \relax \else
      \ifdim\lastskip<1.5em \hskip-\lastskip
      \hskip1.5em plus0em minus0.5em \fi \nobreak
      \vrule height0.75em width0.5em depth0.25em\fi}

\usepackage[english]{babel}
\usepackage{amsmath,amsthm}
\usepackage{amsfonts}
\usepackage{graphicx}
\usepackage{verbatim}
\usepackage{amssymb}
\usepackage{epstopdf}
\usepackage{hyperref}
\usepackage{bbm}
\usepackage{cite,bbm,graphicx,amsmath,amssymb,mathrsfs,epsf} \usepackage{multirow}
\usepackage{subfigure,cite,graphicx,epsfig,amsmath,amssymb,mathrsfs,epsf,graphics,color,enumerate}
\usepackage{color}

\def \E{\operatorname{E}}

\def \var{\operatorname{Var}}

 \allowdisplaybreaks[4]

\begin{document}
\title{ Covert Communication with Channel-State Information at the Transmitter}%

\author{\IEEEauthorblockN{%
 Si-Hyeon Lee, Ligong Wang,  Ashish Khisti, and Gregory W. Wornell}\thanks{S.-H. Lee is with the Department of Electrical Engineering, Pohang University of Science and Technology (POSTECH), Pohang, South Korea 37673 (e-mail: sihyeon@postech.ac.kr). L. Wang is with ETIS--Universit\'e Paris Seine, Universit\'e de Cergy-Pontoise, ENSEA, CNRS, France (e-mail: ligong.wang@ensea.fr).  A. Khisti is with the Department of Electrical and Computer Engineering, University of Toronto, Toronto,  ON M5S, Canada (e-mail: akhisti@comm.utoronto.ca). G. W. Wornell  is with the  Department of Electrical Engineering and Computer Science, Massachusetts Institute of Technology (MIT), Cambridge, MA, USA (e-mail: gww@mit.edu). The material in this paper was  presented in part at IEEE ISIT 2017. }}
 
\maketitle

\begin{abstract}
We consider the problem of covert communication over a state-dependent channel, where the transmitter has causal or noncausal knowledge of the channel states. Here, ``covert'' means that a warden on the channel should observe similar statistics when the transmitter is sending a message and when it is not.  When a sufficiently long secret key is shared between the transmitter and the receiver, we derive closed-form formulas for the maximum achievable covert communication rate (``covert capacity'') for discrete memoryless channels and, when the transmitter's channel-state information (CSI) is noncausal, for additive white Gaussian noise (AWGN) channels. For certain channel models, including the AWGN channel, we show that the covert capacity is positive with CSI at the transmitter, but is zero without CSI.  We also derive lower bounds on the rate of the secret key that is needed for the transmitter and the receiver to achieve the covert capacity.
\end{abstract}

\section{Introduction}
Covert communication \cite{bashgoekeltowsley13,chebakshijaggi13,wangwornellzheng16,bloch16} refers to scenarios where the transmitter and the receiver must keep the warden (eavesdropper) from discovering the fact that they are using the channel to communicate. Specifically, the signals observed by the warden must be statistically close to the signals when the transmitter is switched off. For additive white Gaussian noise (AWGN) channels, the transmitter being switched off is usually modeled by it always sending zero; for discrete memoryless channels (DMCs), this is modeled by it sending a specially designated ``no input'' symbol $x_0$. For a DMC, if the output distribution at the warden generated by $x_0$ is a convex combination of the output distributions generated by the other input symbols, then a positive covert communication rate is achievable; otherwise the maximum amount of information that can be covertly communicated scales like the square root of the total number of channel uses \cite{wangwornellzheng16}. For the AWGN channel, the latter situation applies~\cite{bashgoekeltowsley13,wangwornellzheng16}.

The role played by \emph{channel uncertainties} in covert communications has been studied in some recent works. In particular, \cite{chebakshichanjaggi14,7084182,SobersTowley:arxiv16} consider the situation where the noise level  (or cross-over probability) of the channel is random, remains constant throughout the entire communication duration, and is unknown to the warden. In this case, it is difficult for the warden to tell whether what it observes is signal or noise. As a consequence, positive covert communication rates are achievable on certain channel models (binary symmetric channels are considered in \cite{chebakshichanjaggi14} and AWGN channels in \cite{7084182,SobersTowley:arxiv16}) which, without the unknown-noise-level assumption, only allow square-root scaling for covert communication.

The current work studies the benefit of channel uncertainties for covert communications in a different context. We consider channels with a random state that is independent and identically distributed (IID) across different channel uses. Clearly, if a channel state sequence is not known to any terminal, then it can be treated as part of the channel statistics, reducing the problem to the one studied in previous works. Hence, in general, an IID unknown parameter cannot help the communicating parties to communicate covertly. In the current work, we assume that the state sequence is known to the transmitter, either causally or noncausally, as channel-state information (CSI), but unknown to the receiver and the warden. 
As one motivating application consider a scenario where a noise source or jammer continuously emits IID random noise. If the jammer is friendly and reveals the predetermined noise symbols to the transmitter, then the transmitter can use this knowledge as CSI.
Another scenario is when the path delay from the jammer to the receiver and the warden is larger than the total path delay from the jammer to the transmitter and from the transmitter to the receiver and the warden so that the transmitter knows the jammer's signal in advance and utilizes it as CSI. In the literature, such difference in path delays has motivated the study of lookahead relay channels \cite{4305394,6675764}, \cite[Chapter 16.9]{ElGamalKim:11}. 

%

We study the maximum achievable rate for covert communication, which we call the ``covert capacity,'' in this case. We derive closed-form formulas for the covert capacity, when the transmitter and the receiver share a sufficiently long secret key. We also derive upper bounds on the minimum length of the secret key needed to achieve the covert capacity. We do not have a good lower bound on this  minimum secret-key length; we briefly comment on this in the concluding section. Our converse proofs are based on classical techniques, while covertness is accounted for with help of continuity properties of certain information quantities. Our achievability proof for the noncausal case is based on ``likelihood encoding'' employed in \cite{7707384} rather than standard Gelfand-Pinsker coding \cite{GelfandPinsker:80}, because the former admits easier covertness analysis. For the binary symmetric channel (BSC) and the AWGN channel, in certain parameter ranges, we show the covert capacity to be positive with CSI at the transmitter. (Recall that, without channel state, the covert capacity is zero for both channels in all parameter ranges.)

Our work is closely related to some works in steganography \cite{fridrich09,wangmoulin08,ezzeddinemoulin09}. In steganography, the transmitter is given some data, called the ``cover text,'' and attempts to embed its message by modifying the cover text. As pointed out in \cite{ezzeddinemoulin09}, the cover text can be seen as CSI that is noncausally known to the transmitter. The main difference between such problems and our setting is the following. In steganography it is normally assumed that no noise is imposed on the ``stegotext''---the data after modification by the transmitter, hence, conditional on the states (i.e., the cover text), the channel is noiseless. In our setting, the channel has both states and noise.

The rest of this paper is arranged as follows: Section~\ref{sec:model} formally defines the covert communication problem;  Section~\ref{sec:main} states the main results for DMCs; Sections~\ref{sec:converse} and~\ref{sec:achie} prove the converse and achievability parts of the main results, respectively; Section~\ref{sec:examples} applies the results to BSCs and AWGN channels; and Section~\ref{sec:conclusion} concludes the paper with some remarks.

\section{Problem Formulation} \label{sec:model}
\begin{figure}
\center \includegraphics[width=4in,clip]{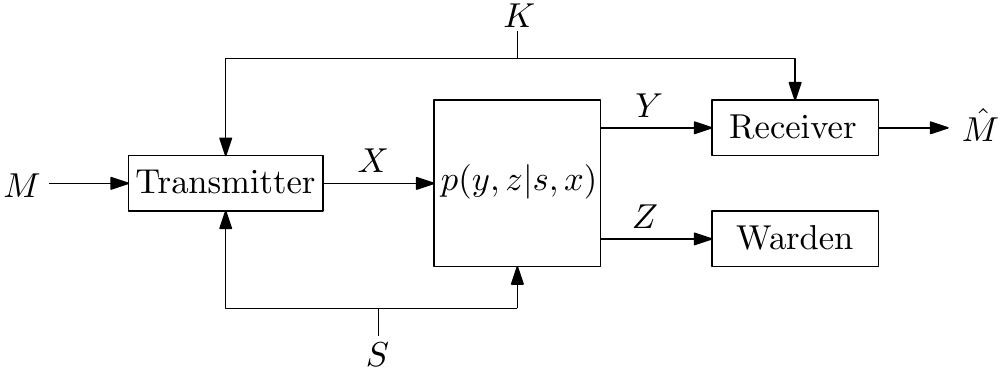}
\caption{State-dependent discrete memoryless channel} \label{fig:model}
\end{figure}

A state-dependent DMC in Fig. \ref{fig:model}
\begin{align}
(\calX, \calS, \calY, \calZ, P_S, P_{Y,Z|S,X})
\end{align}
consists of channel input alphabet $\calX$, state alphabet $\calS$, channel output alphabets $\calY$ and $\calZ$ at the receiver and the warden, respectively, state probability mass function (PMF) $P_S$, and channel law $P_{Y,Z|S,X}$.  
All alphabets are finite. Let $x_0\in \calX$ be a ``no input'' symbol that is sent when no communication takes place. Define $Q_0(\cdot)=\sum_{s\in \calS}P_S(s)P_{Z|S,X}(\cdot|s$, $x_0)$ and let $Q_0^{\times n}(\cdot)$ denote the $n$-fold product of $Q_0$. The state sequence $S^n$ is assumed to be IID, hence the warden observes $Z^n$ distributed according to $Q_0^{\times n}(\cdot)$ if no communication takes place over $n$ channel uses.  
  We define a nonnegative cost $b(x)$ for each input symbol $x\in \calX$. 
The average input cost of $x^n \in \calX^n$ is defined as $b(x^n)=\frac{1}{n}\sum_{i=1}^n b(x_i)$. 

The transmitter and the receiver are assumed to share a secret key $K$ uniformly distributed over a set $\calK$. The state sequence is assumed to be unknown to the receiver and the warden,  but available to the transmitter. We consider two cases, where the state sequence is known to the transmitter causally and noncausally, respectively. For causal CSI,  an $(|\calM|, |\calK|, n)$ code consists of an encoder at the transmitter that maps $(M,K,S^i)$ to $X_i \in \calX$ for $i  \in [1:n]$, and a decoder at the receiver that maps $(Y^n, K)$  to $\hat{M} \in \calM$. For noncausal CSI, an $(|\calM|, |\calK|, n)$ code consists of an encoder at the transmitter that maps $(M,K,S^n)$ to $X^n\in \calX^n$, and a decoder at the receiver that maps $(Y^n, K)$  to $\hat{M} \in \calM$. 

The transmitter and the receiver aim at constructing a code that is both reliable and covert. As usual, their code is reliable if the probability of error $P_e^{(n)}=P(\hat{M}\neq M)$ is negligible. Their code is covert if it is hard for the warden to determine whether the transmitter is sending a message (hypothesis $H_1$) or not (hypothesis $H_0$). Let $\alpha$ and $\beta$ denote the probabilities of false alarm (accepting $H_1$ when the transmitter is not sending a message) and missed detection (accepting $H_0$ when the transmitter is sending a message), respectively. Note that a blind test satisfies $\alpha+\beta=1$. Let  $\widehat{P}_{Z^n}$ denote the distribution   observed by the warden when the transmitter is sending a message.\footnote{Note that $\widehat{P}_{Z^n}$ depends on the code used for the communication and is in general not IID. } The warden's optimal hypothesis test satisfies $\alpha+\beta\geq 1-\sqrt{D(\widehat{P}_{Z^n}\|Q_0^{\times n})}$ (see \cite{LehmannRomano:05}). Hence, covertness is guaranteed if $D(\widehat{P}_{Z^n}\|Q_0^{\times n})$ is negligible. 
At this point, note that  an input symbol $x$ with $\mbox{supp}(P_Z(\cdot|x))\notin \mbox{supp}(Q_0)$, where $\mbox{supp}$ denotes the support set of a distribution, should not be transmitted with nonzero probability because otherwise $D(\widehat{P}_{Z^n}\|Q_0^{\times n})$ becomes infinity. Hence, by dropping such input symbols, we assume that 
\begin{equation}
\mbox{supp}(Q_0)=\mathcal{Z}. \label{eq:suppQ0}
\end{equation}

Let $\calK=[1:2^{nR_K}]$ and $\calM=[1:2^{nR}]$ for $R_K\geq 0$ and $R\geq 0$.  For given $R_K\geq 0$ and $B\geq 0$, a covert rate of $R$ is said to be achievable if there exists a sequence of $(2^{nR},2^{nR_K}, n)$ codes that simultaneously satisfies the input cost constraint  $\limsup_{n\rightarrow \infty}\E_{M, K, S^n}\left[b(X^n)\right]\leq B$,  reliability constraint  $\lim_{n\rightarrow \infty}P_e^{(n)}=0$, and covertness constraint $\lim_{n\rightarrow \infty } D(\widehat{P}_{Z^n}\|Q_0^{\times n})=0$. The \emph{covert capacity} is defined as the supremum of all achievable covert rates and denoted by $C_{\mathrm{c}}$ and  $C_{\mathrm{nc}}$ for the cases with causal CSI and with noncausal CSI, respectively.
\section{Main Results for DMCs} \label{sec:main}
In this section, we present upper and lower bounds on the covert capacity  of DMCs with causal and with noncausal CSI at the transmitter. The proofs of the upper and lower bounds are provided in Sections \ref{sec:converse} and \ref{sec:achie}, respectively. 

\subsection{Causal CSI at the Transmitter}
\begin{theorem}\label{thm:converse_causal}  
For $R_K\geq 0$ and $B\geq 0$, the covert capacity with causal CSI at the transmitter is upper-bounded as 
\begin{align}
C_{\mathrm{c}}\leq \max I(V;Y)  \label{eqn:main_cv_causal}
\end{align}
where the maximum is over PMF $P_{V}$ and function $x(v,s)$ such that $|\calV|\leq \min\{|\calX|+|\calY|+|\calZ|-2, (|\calX|-1)\cdot |\calS|+1  \}$, $P_Z=Q_0$ and $\E[b(X)]\leq B$.
\end{theorem}
\begin{theorem} \label{thm:achie_causal}
For $R_K\geq 0$ and $B\geq 0$, the covert capacity with causal CSI at the transmitter  is lower-bounded as  
\begin{align}
C_{\mathrm{c}}&\geq \max  I(V;Y) \label{eqn:main_rate_causal}
\end{align} 
where the maximum is over  PMF $P_{V}$ and function $x(v,s)$ such that $|\calV|\leq \min\{|\calX|+|\calY|+|\calZ|-1, (|\calX|-1)\cdot |\calS|+2\}$, $P_Z=Q_0$, $\E[b(X)]\leq B$, and 
\begin{align}
I(V;Z)-I(V;Y) < R_K. \label{eqn:main_key_causal}
\end{align}
\end{theorem}

\subsection{Noncausal CSI at the Transmitter }
\begin{theorem}\label{thm:converse}  
For $R_K\geq 0$ and $B\geq 0$, the covert capacity with noncausal CSI at the transmitter  is upper-bounded as 
\begin{align}
C_{\mathrm{nc}}\leq \max (I(U;Y)-I(U;S)) \label{eqn:main_cv}
\end{align}
where the maximum is over conditional PMF $P_{U|S}$ and function $x(u,s)$ such that $|\calU|\leq \min\{|\calX|+|\calY|+|\calZ|+|\calS|-3, |\calX|\cdot |\calS|\}$, $P_Z=Q_0$ and $\E[b(X)]\leq B$.
\end{theorem}
\begin{theorem} \label{thm:achie}
For $R_K\geq 0$ and $B\geq 0$, the covert capacity with noncausal CSI at the transmitter  is lower-bounded as  
\begin{align}
C_{\mathrm{nc}}&\geq \max (I(U;Y)-I(U;S)) \label{eqn:main_rate}
\end{align} 
where the maximum is over conditional PMF $P_{U|S}$ and function $x(u,s)$ such that $|\calU|\leq \min\{|\calX|+|\calY|+|\calZ|+|\calS|-2, |\calX|\cdot |\calS|+1\}$, $P_Z=Q_0$, $\E[b(X)]\leq B$, and 
\begin{align}
I(U;Z)-I(U;Y) < R_K. \label{eqn:main_key}
\end{align}
\end{theorem}
\begin{remark} If we restrict $U$ and $S$ to be independent in Theorems \ref{thm:converse} and \ref{thm:achie}, the bounds fall back to those in Theorems \ref{thm:converse_causal} and \ref{thm:achie_causal}.  
\end{remark} 
\begin{remark}
For the case with causal CSI (resp. noncausal CSI), if $R_K$ is large enough so that \eqref{eqn:main_key_causal} (resp. \eqref{eqn:main_key}) holds under the joint distribution that achieves the maximum on the right-hand side of \eqref{eqn:main_cv_causal} (resp. \eqref{eqn:main_cv}), then Theorems~\ref{thm:converse_causal} and~\ref{thm:achie_causal} (resp. Theorems~\ref{thm:converse} and~\ref{thm:achie}) establish the covert capacity as the right-hand side of \eqref{eqn:main_cv_causal} (resp. \eqref{eqn:main_cv}). Furthermore, if under this joint distribution $I(V;Z)<I(V;Y)$ (resp. $I(U;Z) < I(U;Y)$), then no secret key is needed to achieve the covert capacity.
\end{remark}

\begin{remark}\label{rmk:nogain}

Let us consider the special case of  $Y=Z$ as in \cite{wangwornellzheng16}. Then, in the absence of CSI, the covert capacity can be positive  if and only if $x_0$ is redundant \cite{wangwornellzheng16}, i.e.,  $P_{Z|X}(\cdot|x_0)  \in \mbox{conv}\{P_{Z|X}(\cdot|x')\colon x'\in \mathcal{X}, x'\neq x_0\}$ where $\mbox{conv}$ denotes the convex hull. 
In the presence of CSI, the covert capacity can be positive even though $x_0$ is not redundant. Examples include channels with additive state where some fraction of state can be subtracted through appropriate precoding so that the transmitter can send message symbol (corresponding to $V$ for the causal case and to $U$ for the noncausal case) generated by taking into account the effect of subtracted state. 
In Section \ref{sec:examples}, we show such examples. 
\end{remark}


\section{Proof of Upper Bounds} \label{sec:converse}
In this section, we prove the converse part of our main results, i.e., Theorems \ref{thm:converse_causal} and \ref{thm:converse}. Let us first define the following functions of nonnegative $A$ and $B$: 
\begin{align*}
C_{\mathrm{c}}(A,B)
&=\max_{\substack{P_{V},~x(v,s)\colon \\ \E[b(X)]\leq B,~ D(P_{Z}\|Q_0)\leq A} } I(V;Y) \\
C_{\mathrm{nc}}(A,B)
&=\max_{\substack{P_{U|S},~P_{X|U,S}\colon \\ \E[b(X)]\leq B,~ D(P_{Z}\|Q_0)\leq A} } (I(U;Y)-I(U;S)). 
\end{align*}
In the proof of Theorems \ref{thm:converse_causal} and \ref{thm:converse}, we use the following lemma, which is proven at the end of this section.
\begin{lemma}\label{cl:Cl}
The functions $C_{\mathrm{c}}(A,B)$ and  $C_{\mathrm{nc}}(A,B)$ are non-decreasing in each of $A$ and $B$, and concave and continuous in $(A,B)$. 
\end{lemma} 
Now we are ready to prove Theorems \ref{thm:converse_causal} and \ref{thm:converse}.
\begin{proof}[Proof of Theorem \ref{thm:converse_causal}] For $R_K\geq 0$ and $B\geq 0$, consider any sequence of $(2^{nR},2^{nR_K}, n)$ codes that simultaneously satisfies the input cost constraint $\limsup_{n\rightarrow \infty}\E_{M, K, S^n}\left[b(X^n)\right]\leq B$,  reliability constraint  $\lim_{n\rightarrow \infty}P_e^{(n)}=0$, and covertness constraint $\lim_{n\rightarrow \infty } D(\widehat{P}_{Z^n}\|Q_0^{\times n})=0$. 

Let us start with the proof steps used for channels with causal CSI \cite{ElGamalKim:11} without a covertness constraint: 
\begin{align}
nR&\overset{(a)}\leq I(M;Y^n|K)+n\epsilon_n\\
&=\sum_{i=1}^n I(M;Y_i|K,Y^{i-1})+n\epsilon_n\\
&\leq \sum_{i=1}^n I(M, K, Y^{i-1};Y_i)+n\epsilon_n\\
&\leq \sum_{i=1}^n I(M,K, Y^{i-1},S^{i-1};Y_i)+n\epsilon_n\\
&\overset{(b)}=\sum_{i=1}^n I(M,K,Y^{i-1},S^{i-1}, X^{i-1};Y_i)+n\epsilon_n\\
&\overset{(c)}=\sum_{i=1}^n I(M,K,S^{i-1}, X^{i-1};Y_i)+n\epsilon_n\\
&\overset{(d)}=\sum_{i=1}^n I(V_i;Y_i) +n\epsilon_n  \label{eqn:GP_conv_causal}
\end{align}
for $\epsilon_n\rightarrow 0$ and $V_i:=(M,K,S^{i-1})$. Here, $(a)$ follows by applying Fano's inequality from the reliability constraint; $(b)$ and $(d)$ because $X^{i-1}$ is a function of $(M, K, S^{i-1})$; and $(c)$ since $Y^{i-1}-(M,K,S^{i-1}, X^{i-1})-Y_i$ forms a Markov chain.

Now we utilize the definition and the property of $C_{\mathrm{c}}(A,B)$ to further bound the right-hand side of~\eqref{eqn:GP_conv_causal}: 
\begin{align}
nR&\leq \sum_{i=1}^n I(V_i;Y_i) +n\epsilon_n \\
&\overset{(a)}\leq \sum_{i=1}^n C_{\mathrm{c}}(D(\widehat{P}_{Z_i}\|Q_0), \E[b(X_i)]) + n\epsilon_n\\
&\overset{(b)}\leq nC_{\mathrm{c}}\left(\frac{1}{n}\sum_{i=1}^nD(\widehat{P}_{Z_i}\|Q_0),\frac{1}{n}\sum_{i=1}^n\E[b(X_i)]\right)+n\epsilon_n  \label{eqn:13c}
\end{align}
where $(a)$ is because $X_i$ is a function of $V_i$ and $S_i$ and the Markov chain $V_i-(X_i, S_i)-(Y_i, Z_i)$ holds and $(b)$ is due to the concavity of $C_{\mathrm{c}}(A,B)$. Recall from Lemma~\ref{cl:Cl} that $C_{\mathrm{c}}(A,B)$ is non-decreasing in each of $A$ and $B$. According to the input cost constraint, there exists $\delta_n\rightarrow 0$ such that  $\frac{1}{n}\sum_{i=1}^n\E[b(X_i)]\leq B+\delta_n$. On the other hand, from the covertness constraint, there exists $\delta_n'\rightarrow 0$ such that  $ D(\widehat{P}_{Z^n}\|Q_0^{\times n})$ $\leq \delta_n'$, while  as in \cite{wangwornellzheng16} we have
\begin{align}
&D(\widehat{P}_{Z^n}\|Q_0^{\times n})\nonumber \\*
&=-H(Z^n)+\E_{\widehat{P}_{Z^n}}\left[\log \frac{1}{Q_0^{\times n}(Z^n)}\right]\label{eqn:ccf}\\
&=-\sum_{i=1}^nH(Z_i|Z^{i-1})+\E_{\widehat{P}_{Z^n}}\left[\log \frac{1}{Q_0(Z_i)}\right]\\
&=-\sum_{i=1}^nH(Z_i|Z^{i-1})+\E_{\widehat{P}_{Z_i}}\left[\log \frac{1}{Q_0(Z_i)}\right]\\
&\geq -\sum_{i=1}^nH(Z_i)+\E_{\widehat{P}_{Z_i}}\left[\log \frac{1}{Q_0(Z_i)}\right]\\
&=\sum_{i=1}^nD(\widehat{P}_{Z_i}\|Q_0).\label{eqn:ccl}
\end{align}
Hence, \eqref{eqn:13c} implies
\begin{align}
R\leq C_{\mathrm{c}}\left(\frac{\delta_n'}{n},B+\delta_n\right)+\epsilon_n. \label{eqn:conv_no_causal}
\end{align}
Note that the right-hand side of \eqref{eqn:conv_no_causal} approaches $C_{\mathrm{c}}(0,B)$ as $n$ tends to infinity due to the continuity of $C_{\mathrm{c}}(A,B)$, from which follows the condition $P_Z=Q_0$. Finally, the cardinality bound on $\calU$ follows by applying the support lemma \cite{ElGamalKim:11}. 
\end{proof}

\begin{proof}[Proof of Theorem \ref{thm:converse}] For $R_K\geq 0$ and $B\geq 0$, consider any sequence of $(2^{nR},2^{nR_K}, n)$ codes that simultaneously satisfies the input cost constraint $\limsup_{n\rightarrow \infty}\E_{M, K, S^n}\left[b(X^n)\right]\leq B$,  reliability constraint  $\lim_{n\rightarrow \infty}P_e^{(n)}=0$, and covertness constraint $\lim_{n\rightarrow \infty } D(\widehat{P}_{Z^n}\|Q_0^{\times n})=0$. 

We start with the proof steps used for channels with noncausal CSI \cite{GelfandPinsker:80} without a covertness constraint: 
\begin{align}
nR&\overset{(a)}\leq I(M;Y^n|K)+n\epsilon_n\\
&=\sum_{i=1}^n I(M;Y_i|K,Y^{i-1})+n\epsilon_n\\
&\leq \sum_{i=1}^n I(M, K, Y^{i-1};Y_i)+n\epsilon_n\\
&=\sum_{i=1}^n I(M,K, Y^{i-1},S_{i+1}^n;Y_i)\nonumber \\*
&\qquad -\sum_{i=1}^n I(Y_i;S_{i+1}^n |M,K, Y^{i-1})+n\epsilon_n\\
&\overset{(b)}=\sum_{i=1}^n I(M,K,Y^{i-1},S_{i+1}^n;Y_i)\nonumber \\*
&\qquad -\sum_{i=1}^n I(Y^{i-1};S_i|M, K, S_{i+1}^n) + n\epsilon_n\\
&\overset{(c)}=\sum_{i=1}^n I(M,K, Y^{i-1},S_{i+1}^n;Y_i)\nonumber \\*
&\qquad -\sum_{i=1}^n I(M,K,Y^{i-1},S_{i+1}^n;S_i) + n\epsilon_n\\
&=\sum_{i=1}^n(I(U_i;Y_i)-I(U_i;S_i))+n\epsilon_n  \label{eqn:GP_conv}
\end{align}
for $\epsilon_n\rightarrow 0$ and $U_i:=(M,K,Y^{i-1}, S_{i+1}^n)$. Here, $(a)$ follows by applying Fano's inequality from the reliability constraint; $(b)$ by Csisz\'ar's sum identity; and $(c)$ because $S_i$ and $(M,K, S_{i+1}^n)$ are independent. 

Now we utilize the definition and the property of $C_{\mathrm{nc}}(A,B)$ to further bound the right-hand side of \eqref{eqn:GP_conv}: 
\begin{align}
nR&\leq \sum_{i=1}^n(I(U_i;Y_i)-I(U_i;S_i))+n\epsilon_n \\
&\leq \sum_{i=1}^n C_{\mathrm{nc}}(D(\widehat{P}_{Z_i}\|Q_0), \E[b(X_i)]) + n\epsilon_n\\
&\overset{(a)}\leq nC_{\mathrm{nc}}\left(\frac{1}{n}\sum_{i=1}^nD(\widehat{P}_{Z_i}\|Q_0),\frac{1}{n}\sum_{i=1}^n\E[b(X_i)]\right)+n\epsilon_n  \label{eqn:13}
\end{align}
where $(a)$ is due to the concavity of $C_{\mathrm{nc}}(A,B)$. Recall from Lemma~\ref{cl:Cl} that $C_{\mathrm{nc}}(A,B)$ is non-decreasing in each of $A$ and $B$. According to the input cost constraint, there exists $\delta_n\rightarrow 0$ such that  $\frac{1}{n}\sum_{i=1}^n\E[b(X_i)]\leq B+\delta_n$. On the other hand, due to the chain of inequalities \eqref{eqn:ccf}-\eqref{eqn:ccl} from the covertness constraint, there exists $\delta_n'\rightarrow 0$ such that  $\sum_{i=1}^nD(\widehat{P}_{Z_i}\|Q_0)\leq \delta_n'$. Hence, \eqref{eqn:13} implies
\begin{align}
R\leq C_{\mathrm{nc}}\left(\frac{\delta_n'}{n},B+\delta_n\right)+\epsilon_n. \label{eqn:conv_no_n}
\end{align}
Note that the right-hand side of \eqref{eqn:conv_no_n} approaches $C(0,B)$ as $n$ tends to infinity due to the continuity of $C_{\mathrm{nc}}(A,B)$, from which follows the condition $P_Z=Q_0$. Further, because $I(U;Y)-I(U;S)$ is convex in the conditional distribution $P_{X|U,S}$, it suffices to maximize it over functions $x(u,s)$ instead of $P_{X|S,U}$. Finally, the cardinality bound on $\calU$ follows by applying the support lemma \cite{ElGamalKim:11}. 
\end{proof}

\begin{proof}[Proof of Lemma \ref{cl:Cl}]
Let us show that  $C_{\mathrm{nc}}(A,B)$ is non-decreasing in each of $A$ and $B$, and concave and continuous in $(A,B)$. It can be proved in a similar manner that the same statement holds for  $C_{\mathrm{c}}(A,B)$.

First,  $C_{\mathrm{nc}}(A,B)$ is non-decreasing in each of $A$ and $B$ since increasing $A$ or $B$ can only enlarge the set of feasible $P_{U|S}P_{X|U,S}$. 

Second, to show the concavity, fix arbitrary $(A_1, B_1)$ and $(A_2, B_2)$ and let $P_{U_1|S}P_{X_1|U_1,S}$ and $P_{U_2|S}$ $P_{X_2|U_2,S}$ denote the corresponding conditional PMFs that achieve the maxima of $C_{\mathrm{nc}}(A_1,B_1)$ and $C_{\mathrm{nc}}(A_2,B_2)$, respectively. Let $Y_1$ and $Z_1$ (resp. $Y_2$ and $Z_2$) denote the channel outputs at the receiver and the warden,  respectively, corresponding to $P_{U_1|S}P_{X_1|U_1,S}$ (resp. $P_{U_2|S}P_{X_2|U_2,S}$). Let $Q$ denote a random variable independent of $U_1, U_2$, and $S$, which takes value $1$ with probability $\lambda$ and value $2$ with probability $1-\lambda$. 

Define $U'=(Q, U_Q)$. Let $X'$, $Y'$, and $Z'$ denote the channel input at the transmitter and the channel outputs at the receiver and the warden, respectively, corresponding to  $P_{X'|U',S}=P_{X_Q|U_Q,S}$. Note that $P_{X'}=\lambda P_{X_1} + (1-\lambda) P_{X_2}$ and $P_{Z'}=\lambda P_{Z_1} + (1-\lambda) P_{Z_2}$.  Then, 
\begin{align}
\E [b(X')]&=\sum_{x\in \calX}P_{X'}(x)b(x)\nonumber \\*
&=\lambda \sum_{x\in \calX}P_{X_1}(x)b(x)+(1-\lambda) \sum_{x\in \calX}P_{X_2}(x)b(x)\nonumber \\*
&\leq \lambda B_1 +(1-\lambda) B_2
\end{align}
and
\begin{align}
D(P_{Z'}\|Q_0)&\leq \lambda D(P_{Z_1}\|Q_0) + (1-\lambda) D(P_{Z_2}\|Q_0)\nonumber \\*
&\leq \lambda A_1 +(1- \lambda) A_2
\end{align}
because the relative entropy is convex in the first argument. 
Hence, it follows that 
\begin{align}
C_{\mathrm{nc}}(\lambda (A_1,B_1) + (1-\lambda) (A_2,B_2))&\geq I(U';Y')-I(U';S)  \\
&\geq I(U_Q, Q;Y')-I(U_Q,Q;S)\\
&\overset{(a)}\geq I(U_Q;Y'|Q)-I(U_Q;S|Q)\\
&=\lambda I(U_1;Y_1) + (1-\lambda) I(U_2;Y_2) \nonumber \\* 
&\qquad - \lambda I(U_1;S) -(1-\lambda)I(U_2;S)\\
&=\lambda C_{\mathrm{nc}}(A_1,B_1) + (1-\lambda) C_{\mathrm{nc}}(A_2,B_2),
\end{align} 
where $(a)$ is because $Q$ and $S$ are independent. Hence, we conclude that $C_{\mathrm{nc}}(A,B)$ is concave.

Lastly, $C_{\mathrm{nc}}(A,B)$ is continuous in $(A,B)$ since both the objective function $I(U;Y)-I(U;S)$ and the constraint functions $\E[b(X)]$ and $D(P_Z\|Q_0)$ are continuous in $P_{U|S}P_{X|U,S}$ as long as  \eqref{eq:suppQ0} is satisfied.
\end{proof}

\section{Proof of Lower Bounds} \label{sec:achie}
In this section, we prove the achievability parts of our main results, i.e., Theorems \ref{thm:achie_causal} and \ref{thm:achie}. To prove the achievability part for the case with causal CSI at the transmitter, we employ the Shannon's strategy \cite{Shannon:58} and use the soft covering theorem  \cite[Theorem 4]{256486},\cite[Corollary VII.4]{6584816} for the covertness analysis. For the case with noncausal CSI at the transmitter, our scheme is based on  multicoding \cite{GelfandPinsker:80, ElGamalKim:11}, but instead of performing the joint-typicality check to find a codeword that \emph{seemingly} follows a joint distribution with the state sequence, we use likelihood encoding employed in \cite{7707384} since it admits easier covertness analysis.

\begin{proof}[Proof of Theorem \ref{thm:achie_causal}] 
Fix $\epsilon>0$. Further fix $P_{V}$ and  $x(v,s)$ such that $P_Z=Q_0$ and  $\E[b(X)]\leq \frac{B}{1+\epsilon}$.
\subsubsection{Codebook generation}    For each $k\in [1:2^{nR_K}]$ and $m\in [1:2^{nR}]$, randomly and independently generate a $v^n(k,m)$ according to $\prod_{i=1}^nP_V(v_i)$. These constitute the codebook $\calC$. 

\subsubsection{Encoding at the transmitter} Given state sequence $s^n$, secret key $k$, and message $m$, the encoder transmits $x^n$ where $x_i=x(v_i(k,m),s_i)$. 

\subsubsection{Decoding at the receiver} Upon receiving $y^n$,  with access to the secret key $k$, the decoder declares that  $\hat{m}$ is sent if it is the unique message such that 
\begin{align}
(v^n(k, \hat{m}), y^n)\in \calT^{(n)}_{\epsilon}. 
\end{align}
Otherwise it declares an error. Here $\calT^{(n)}_\epsilon$ denotes the (strongly) typical set \cite{CsiszarKorner:11}.

\subsubsection{Covertness analysis}
By the soft covering theorem \cite[Theorem 4]{256486},\cite[Corollary VII.4]{6584816}, we have  $\E_{\calC} [D(\widehat{P}_{Z^n}\|Q_0^{\times n})] \overset{n\rightarrow \infty}\longrightarrow 0$ if 
\begin{align}
R+R_K&>I(V;Z).\label{eqn:scover}
\end{align}

\subsubsection{Reliability and input cost analysis} 
By the standard error analysis, it can be shown that the probability of error averaged over the random codebook $\calC$ tends to zero as $n$ tends to infinity if
\begin{align}
R<I(V;Y). \label{eqn:pak}
\end{align}
Furthermore, by the typical average lemma \cite{ElGamalKim:11},  
\begin{align}
\E_{\mathcal{C}, M, K, S^n}\left[b(X^n)\right]\nonumber &=P(X^n\notin T_{\epsilon}^{(n)}) \E_{\mathcal{C}, M, K, S^n}\left[b(X^n)|X^n\notin T_{\epsilon}^{(n)}\right]\nonumber \\*
&\qquad +P(X^n\in T_{\epsilon}^{(n)})\E_{\mathcal{C}, M, K, S^n}\left[ b(X^n)|X^n\in T_{\epsilon}^{(n)}\right]\\
&\leq P(X^n\notin T_{\epsilon}^{(n)}) B_{\max} + B, 
\end{align}
where $B_{\max}:=\max_{x\in \calX}b(x)$. Note that $P(X^n\notin T_{\epsilon}^{(n)})\rightarrow 0$ as $n$ tends to infinity. Hence, we have
\begin{align}
\limsup_{n\rightarrow \infty}\E_{\mathcal{C}, M, K, S^n}\left[b(X^n)\right]\leq B.
\end{align}

In summary,  if \eqref{eqn:scover} and \eqref{eqn:pak}  are satisfied, then there must exist a sequence of codes such that $\lim_{n\rightarrow \infty}P_e^{(n)}=0$, $\lim_{n\rightarrow \infty}\E_{M, K, S^n}\left[b(X^n)\right]\leq B$,   and $\lim_{n\rightarrow \infty } D(P_{Z^n}\|Q_0^{\times n})=0$. By applying the Fourier-Mozkin elimination \cite{ElGamalKim:11} to  \eqref{eqn:scover} and \eqref{eqn:pak}, we complete the proof. \end{proof}

\begin{proof}[Proof of Theorem \ref{thm:achie}] 
Fix $\epsilon>\epsilon'>0$.
Further fix $P_{U|S}$ and  $x(u,s)$ such that $P_Z=Q_0$ and  $\E[b(X)]\leq \frac{B}{1+\epsilon'}$.

\setcounter{subsubsection}{0} 
\subsubsection{Codebook generation}    For each $k\in [1:2^{nR_K}]$ and $m\in [1:2^{nR}]$, randomly and independently generate $2^{nR'}$ codewords $u^n(k,m,l)$,  $l\in [1:2^{nR'}]$ according to $\prod_{i=1}^nP_U(u_i)$. These constitute the codebook $\calC$. \label{subsubsec:codebook}

\subsubsection{Encoding at the transmitter} Given state sequence $s^n$, secret key $k$, and message $m$, evaluate the likelihood 
\begin{align}
g(l|s^n,k,m)=\frac{P^{\times n}_{S|U}(s^n|u^n(k,m,l))}{\sum_{l'\in [1:2^{nR'}]}P^{\times n}_{S|U}(s^n|u^n(k,m,l'))}. \label{eqn:likeli}
\end{align}
The encoder randomly generates $l$ according to \eqref{eqn:likeli} and  transmits $x^n$ where $x_i=x(u_i(k,m,l),s_i)$. 
\label{subsubsec:encoding}

\subsubsection{Decoding at the receiver} Upon receiving $y^n$,  with access to the secret key $k$, the decoder declares that  $\hat{m}$ is sent if it is the unique message such that 
\begin{align}
(u^n(k, \hat{m}, l), y^n)\in \calT^{(n)}_{\epsilon} \label{eqn:typical}
\end{align}
for some $l\in [1:2^{nR'}]$;  if no such unique $\hat{m}$ can be found, it declares an error. 

\subsubsection{Covertness analysis}
For covertness analysis, we use the following lemma, which is proven  at the end of this section. 
\begin{lemma} \label{lemma:zz}
For the codebook generation and encoding procedure described above, if $R'>I(U;S)$ and $R+R_K+R'>I(U;Z)$, then 
\begin{align}
\E_{\calC} \left[ D(\widehat{P}_{Z^n} \| P^{\times n}_{Z})\right]\overset{n\rightarrow \infty}\longrightarrow 0. \label{eqn:zz}
\end{align}
\end{lemma}
Now, let 
\begin{align}
R'&>I(U;S)\label{eqn:cover1}\\
R+R_K+R'&>I(U;Z).\label{eqn:cover2}
\end{align}
Because $P_{U|S}$ and $x(u,s)$ are chosen to  satisfy $P_Z=Q_0$, Lemma \ref{lemma:zz} implies that  
\begin{align}
\E_{\calC} [D(\widehat{P}_{Z^n}\|Q_0^{\times n})] \overset{n\rightarrow \infty}\longrightarrow 0.
\end{align}

\subsubsection{Reliability analysis} \label{subsubsec:relia}
Consider the probability of error averaged over the randomly generated codebook~$\mathcal{C}$. Let $M$ and $\hat{M}$ denote the transmitted and decoded messages, respectively, and let $L$ denote the index generated according to \eqref{eqn:likeli} at the encoder. The error event $\{\hat{M} \neq M\}$ occurs only if at least one of the following events occurs: 
\begin{align}
\calE_1&:=\{ (U^n(K,M,L), S^n)\notin \calT_{\epsilon'}^{(n)}\} \label{eqn:e1}\\
\calE_2&:=\{ (U^n(K,M,L), Y^n)\notin \calT_{\epsilon}^{(n)}\} \label{eqn:e2} \\
\calE_3&:=\{ (U^n(K,m,l), Y^n)\in \calT_{\epsilon}^{(n)}  \mbox { for some } m\neq M \mbox{ and } l\in [1:2^{nR'}]\}. \label{eqn:e3}
\end{align}
Hence, the probability of error is bounded as 
\begin{align}
P(\hat{M}\neq M)\leq P(\calE_1)+P(\calE_1^c\cap \calE_2) + P(\calE_3). \label{eqn:prob_ub}
\end{align} 
Now we bound each term on the right-hand side of \eqref{eqn:prob_ub}. The first term $P(\calE_1)$ tends to zero as $n$ tends to infinity due to \cite[Lemma 2]{7707384}, as long as \eqref{eqn:cover1} is satisfied.
Next, note that 
\begin{align}
\calE_1^c&=\{(U^n(K,M,L),S^n)\in  \calT_{\epsilon'}^{(n)}\}.
\end{align} 
By the conditional typicality lemma \cite{ElGamalKim:11}, $P(\calE_1^c\cap \calE_2)$ tends to zero as $n$ tends to infinity. Lastly, $P(\calE_3)$ tends to zero as $n$ tends to infinity by the packing lemma \cite{ElGamalKim:11} provided
\begin{align}
R+R'&<I(U;Y). \label{eqn:4rel}
\end{align}
In summary, the probability of error averaged over the random codebook $\calC$ tends to zero as $n$ tends to infinity if \eqref{eqn:cover1}, \eqref{eqn:cover2}, and \eqref{eqn:4rel} are satisfied.

\subsubsection{Input cost analysis} In the reliability analysis, it is shown that 
\begin{align}
P(\calE_1)&=P\{(U^n(K,M,L),S^n)\notin  \calT_{\epsilon'}^{(n)}\}\\
&=P((U^n(K,M,L),X^n, S^n)\notin  \calT_{\epsilon'}^{(n)}) \overset{n\rightarrow \infty}{\longrightarrow} 0. \label{eqn:typ_cost}
\end{align}
Note that if $x^n \in \calT_{\epsilon'}^{(n)}$, then $b(x^n)\leq B$ by the typical average lemma \cite{ElGamalKim:11}. Hence, 
\begin{align}
&\E_{\mathcal{C}, M, K, S^n}\left[b(X^n)\right]\nonumber \\*
&=P(\calE_1) \E_{\mathcal{C}, M, K, S^n}\left[b(X^n)|\calE_1\right]\nonumber \\*
&\qquad +P(\calE_1^c)\E_{\mathcal{C}, M, K, S^n}\left[ b(X^n)|\calE_1^c\right]\\
&\leq P(\calE_1) B_{\max} + P(\calE_1^c) B, \label{eqn:pow}
\end{align}
where $B_{\max}:=\max_{x\in \calX}b(x)$. By \eqref{eqn:typ_cost}, the right-hand side of \eqref{eqn:pow} approaches $B$ as $n$ tends to infinity. Hence, we have
\begin{align}
\limsup_{n\rightarrow \infty}\E_{\mathcal{C}, M, K, S^n}\left[b(X^n)\right]\leq B.
\end{align}

In summary,  if \eqref{eqn:cover1}, \eqref{eqn:cover2}, and \eqref{eqn:4rel} are satisfied, then there must exist a sequence of codes such that $\lim_{n\rightarrow \infty}P_e^{(n)}=0$, $\lim_{n\rightarrow \infty}\E_{M, K, S^n}\left[b(X^n)\right]\leq B$,   and $\lim_{n\rightarrow \infty } D(P_{Z^n}\|Q_0^{\times n})=0$. By applying the Fourier-Mozkin elimination \cite{ElGamalKim:11} to \eqref{eqn:cover1}, \eqref{eqn:cover2}, and \eqref{eqn:4rel}, we complete the proof. 
\end{proof}

\begin{remark}
We note that our scheme for the case with noncausal CSI at the transmitter is similar to that in \cite{GoldfeldCuffPermuter:arxiv16} for wiretap channels with noncausal CSI at the transmitter under the semantic-security metric requiring negligible information leakage for all message distributions. The coding scheme in  \cite{GoldfeldCuffPermuter:arxiv16} incorporates superposition coding; the inner codebook is for a random index and the outer codebook is for the message and another random index. To compare  that scheme with ours, let us consider the special case of the scheme  in \cite{GoldfeldCuffPermuter:arxiv16} where the rate of the inner codebook is set to zero and let $R$ and $R'$ denote the rates of the message and the random index in the outer codebook. For our scheme, let us consider the special case of  $R_K=0$ and sufficiently large $B$, i.e., no secret key and no input cost constraint.  Then the codebook generation and the encoding procedure of the scheme \cite{GoldfeldCuffPermuter:arxiv16} become the same as our scheme. As shown in the proof of Theorem~\ref{thm:achie}, reliability at the  receiver is satisfied if $R'>I(U;S)$ and $R+R'<I(U;Y)$, but covertness requires $R+R'>I(U;Z)$ and $P_Z=Q_0$ while semantic security requires $R'>I(U;Z)$. 
\end{remark} 

\begin{IEEEproof}[Proof of Lemma~\ref{lemma:zz}]
The proof follows similar lines to \cite[Section VII-A]{7707384}. As in \cite[Section VII-A]{7707384}, it can be checked that, to prove \eqref{eqn:zz}, it suffices to show that the total variation (TV) distance approaches zero:
\begin{align}
\E_{\calC}\|\widehat{P}_{Z^n}-P^{\times n}_{Z}\|_{\mathrm{TV}}\overset{n\rightarrow \infty} \longrightarrow 0. \label{eqn:ZZ_TV}
\end{align}
To evaluate the TV distance, define the ideal PMF for codebook $\calC$ as follows:
\begin{align*}
&\Gamma^{(\calC)}(k,m,l,u^n, s^n, z^n)=\nonumber \\*
&2^{-n(R_K+R+R')} \mathbbm{1}_{u^n(k,m,l)=u^n}P_{S|U}^{\times n} (s^n|u^n) P^{\times n}_{Z|U,S}(z^n|u^n,s^n).
\end{align*}
Using the triangle inequality for the TV distance, we upper-bound the left-hand side of \eqref{eqn:ZZ_TV} as 
\begin{align}
\E_{\calC}\|\widehat{P}_{Z^n}-P^{\times n}_{Z}\|_{\mathrm{TV}}&\leq \E_{\calC}\|\widehat{P}_{Z^n}-\Gamma_{Z^n}^{(\calC)}\|_{\mathrm{TV}}+\E_{\calC}\|\Gamma_{Z^n}^{(\calC)}-P^{\times n}_{Z}\|_{\mathrm{TV}}. \label{eqn:TV_two}
\end{align}
From the soft covering theorem \cite[Theorem 4]{256486},\cite[Corollary VII.4]{6584816}, the second term on the right-hand side of \eqref{eqn:TV_two} decays to zero as $n\rightarrow \infty$ if $R_K+R+R'>I(U;Z)$. 
For the first term on the right-hand side of \eqref{eqn:TV_two}, note that 
\begin{align}
\E_{\calC}\|\widehat{P}_{Z^n}-\Gamma_{Z^n}^{(\calC)}\|_{\mathrm{TV}}\leq \E_{\calC}\|\widehat{P}_{S^n, Z^n}-\Gamma_{S^n,Z^n}^{(\calC)}\|_{\mathrm{TV}}. \label{eqn:TV_spr}
\end{align}
By applying the same analysis as in \cite[Section VII-A]{7707384}, the right-hand side of \eqref{eqn:TV_spr} decays to zero as $n\rightarrow \infty$ if $R'>I(U;S)$. 
\end{IEEEproof}

\section{Examples} \label{sec:examples}
In this section, we show two examples where the covert capacity of a channel is zero in the absence of  CSI at the transmitter, but is positive with CSI.
\subsection{ The Binary Symmetric Channel}  
\begin{figure}
\center \includegraphics[width=2in,clip]{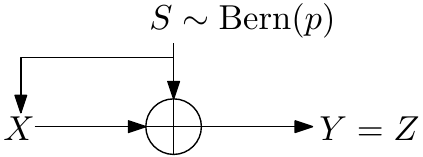}
\caption{Binary symmetric channel with CSI at the transmitter} \label{fig:bsc_ex}
\end{figure}

Consider a channel in Fig. \ref{fig:bsc_ex} where $\calX$, $\calY$, $\calZ$, and $\calS$ are all binary, and where $P_S$ is the Bernoulli distribution of parameter $p\in(0,0.5)$. The channel law is 
\begin{equation}
Y=Z = X\oplus S.
\end{equation}
Assume that $x_0=0$ and $R_K>0$.

Using Theorems~\ref{thm:converse_causal} and~\ref{thm:achie_causal} one can check that,  with causal CSI, the optimal choice is $V=Y=Z$ having the Bernoulli distribution of parameter $p$. This gives
\begin{equation}
C_{\mathrm{c}}= H_\textnormal{b}(p) = p\log\frac{1}{p}+(1-p)\log\frac{1}{1-p}.
\end{equation}
Furthermore, it can be checked that $C_{\mathrm{nc}}=C_{\mathrm{c}}$. 
Note that, without CSI, covert communication cannot have a positive rate \cite{chebakshijaggi13,wangwornellzheng16}   on this channel.

\subsection{  The AWGN Channel}\label{subsec:gaussian}
\begin{figure}
\center \includegraphics[width=3.6in,clip]{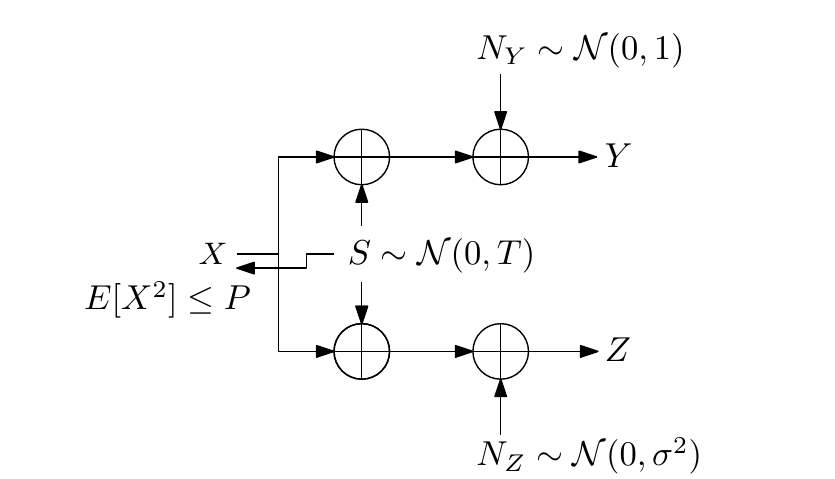}
\caption{AWGN channel with CSI at the transmitter} \label{fig:Gaussian_ex}
\end{figure}
Consider an AWGN channel in Fig. \ref{fig:Gaussian_ex} where the channel outputs at the receiver and the warden are given as 
\begin{align}
Y&=X+S+N_Y \label{eqn:ch_Y} \\
Z&=X+S+N_Z,
\end{align} 
respectively, where $X$ is the channel input from the transmitter, $S\sim \calN(0,T)$ is the external interference that is known to the transmitter causally or noncausally but unknown to the receiver and the warden, and $N_Y\sim \calN(0,1)$ and $N_Z\sim \calN(0,\sigma^2)$, $\sigma^2>0$, are additive Gaussian noises. Let $P$ denote the input power constraint at the transmitter, so the input must satisfy $\E[X^2]\le P$. The ``no input" symbol is $0$, hence the warden observes $Z^n$ distributed according to $Q_0^{\times n}$, where $Q_0=\calN(0,T+\sigma^2)$, when no communication takes place over $n$ channel uses. The transmitter and the receiver are assumed to share a secret key of rate $R_K$. The covertness constraint is again given by $\lim_{n\rightarrow \infty } D(\widehat{P}_{Z^n}\|Q_0^{\times n})=0$. The covert capacity of this channel is defined in the same way as in Section \ref{sec:model} and denoted by $C_{\mathrm{c}}$ and $C_{\mathrm{nc}}$  for causal and noncausal CSI cases, respectively. 

The following theorems show that the covert capacity can be positive for the AWGN channel both with causal CSI and with noncausal CSI at the transmitter. In the following, we define
\begin{subequations}\label{eqn:defs}
\begin{align}
\gamma^* &:= \min\left\{1,\frac{P}{2T}\right\} \label{eqn:gamma}\\ 
T^*& :=(1-\gamma^*)^2T\\
P^* &:= T-T^*.  
\end{align}
\end{subequations}

\begin{theorem} \label{thm:Gaussian_causal}
If 
\begin{align}
R_K&> \frac{1}{2} \log \left(1+\frac{P^*}{T^*+\sigma^2}\right)-\frac{1}{2} \log \left(1+\frac{P^*}{T^*+1}\right),\label{eqn:key_causal}
\end{align}
the covert capacity with causal CSI at the transmitter  is lower-bounded as 
\begin{align}
C_{\mathrm{c}}\geq \frac{1}{2}\log \left(1+\frac{P^*}{T^*+1} \right). \label{eqn:cap_causal}
\end{align}
\end{theorem}

\begin{theorem} \label{thm:Gaussian}
If 
\begin{align}
R_K&> \frac{1}{2}\log \left(1+\frac{(P^*+\frac{P^*}{P^*+1}T^*)^2 }{(P^*+(\frac{P^*}{P^*+1})^2T^*)(P^*+T^*+\sigma^2)-(P^*+\frac{P^*}{P^*+1}T^*)^2 }\right)\nonumber \\
&\quad -\frac{1}{2}\log \left(1+\frac{(P^*+\frac{P^*}{P^*+1}T^*)^2 }{(P^*+(\frac{P^*}{P^*+1})^2T^*)(P^*+T^*+1)-(P^*+\frac{P^*}{P^*+1}T^*)^2 }\right), \label{eqn:key}
\end{align}
 the covert capacity  is given by
\begin{align}
C_{\mathrm{nc}}=\frac{1}{2}\log \left(1+P^*\right). \label{eqn:cap}
\end{align}
\end{theorem}

\begin{remark}
If the warden's channel is degraded, i.e., $\sigma^2> 1 $, a secret key  is not needed to achieve the rates \eqref{eqn:cap_causal} and  \eqref{eqn:cap} for the cases with causal CSI and with noncausal CSI at the transmitter, respectively. 
\end{remark}

\begin{remark}
Let us assume that $R_K$ is sufficiently large so that \eqref{eqn:key_causal} and \eqref{eqn:key} are satisfied. If $T^*=0$, i.e., $T\leq \frac{P}{2}$, then $C_{\mathrm{c}} = C_{\mathrm{nc}}$.
On the other hand, if  $T\to\infty$, it follows   $P^*\to P$. Then, $C_{\mathrm{nc}}$ approaches $\frac{1}{2}\log (1+P)$, which is the capacity of the channel \eqref{eqn:ch_Y} with noncausal CSI at the transmitter and without a covertness constraint. 

\end{remark}


We prove Theorems~\ref{thm:Gaussian_causal} and~\ref{thm:Gaussian} by adapting our DMC results in Theorems~\ref{thm:achie_causal}, \ref{thm:converse}, and~\ref{thm:achie}. In the achievability proofs of Theorems \ref{thm:Gaussian_causal} and \ref{thm:Gaussian},  we reduce the interference power to make room for message transmission. We set the channel input to have the form of $X=X^*-\gamma^* S$ where $X^*$ is independent   of $S$,  so that $\gamma^* S$ is subtracted from $S$ when $X$ is sent. Then, we regard $X^*$ as the input for the channel with reduced interference  power of $(1-\gamma^*)^2T$, i.e., $T^*$. To satisfy the covertness constraint, $X^*$ must have  power $T-T^*=P^*$. Note that the choice of $\gamma^*$ in \eqref{eqn:gamma} ensures that the power constraint of $X$ is satisfied, i.e., 
\begin{align}
\E[X^2]=\E[X^{*2}]+\gamma^{*2}T = T- (1-\gamma^*)^2T+\gamma^{*2}T=2\gamma^*T\leq P.
\end{align}
For the case with causal CSI, the right-hand side of \eqref{eqn:cap_causal} is achieved by letting $V=X^*$ and treating interference as noise at the receiver.  For the case with noncausal CSI, the  right-hand side of \eqref{eqn:cap} is achieved by choosing $U$ as in   ``dirty paper coding''~\cite{1056659}. 

In the following we first prove Theorem~\ref{thm:Gaussian}.

\begin{proof}[Achievability proof of Theorem \ref{thm:Gaussian}] 
We modify the proof of Theorem~\ref{thm:achie} so that it applies to the Gaussian case with a power constraint. Roughly speaking, our idea is to ``quantize'' at the decoder but not at the encoder. 
We choose a conditional probability density function (PDF) of $U$ given $S$ and a mapping from $(U,S)$ to $X$ via the following:
\begin{align}
X^*&\sim \calN(0, P^*), \mbox{ independent of } S\\ 
U&=X^*+ \frac{P^*}{P^*+1}(1-\gamma^*)S \\
X&= U - \frac{P^*+\gamma^*}{P^*+1} S = X^*-\gamma^* S.
\end{align} 
We then employ the same encoding procedure as in Theorem~\ref{thm:achie}, except that PMFs are now replaced by PDFs. Next, as an additional step for the encoder, we fix some small positive $\epsilon$ and check whether the resulting input sequence $x^n$ satisfies the power constraint
\begin{equation}\label{eq:extrapower}
\sum_{i=1}^n x_i^2 \le n (P+\epsilon)
\end{equation}
or not. Denote by $\mathcal{E}_4$ the event that \eqref{eq:extrapower} is \emph{not} satisfied. When $\mathcal{E}_4$ occurs, we replace $x^n$ by the all-zero sequence. By similar analysis as in \cite{7707384} one can show that the probability of $\mathcal{E}_4$ tends to zero as $n$ tends to infinity for all positive $\epsilon$.
Clearly, in the limit where $\epsilon$ approaches zero, our encoding scheme above satisfies the given power constraint.

For covertness analysis, we adapt the proof for the DMC case as follows. Let $\widehat{P}_{Z^n}$ denote the distribution at the warden generated by the above coding scheme, and let $\bar{P}_{Z^n}$ denote the distribution generated by this scheme but \emph{without the additional step} of replacing those codewords not satisfying \eqref{eq:extrapower} with the all-zero sequence. It is clear from the proof of Lemma~\ref{lemma:zz} that it can be applied to PDFs without a maximum-cost constraint, so we can write, similarly to \eqref{eqn:ZZ_TV}, that
\begin{equation}\label{eq:Gausscovert1}
\E_{\calC}\|\bar{P}_{Z^n}-Q_0^{\times n}\|_{\mathrm{TV}}\overset{n\rightarrow \infty} \longrightarrow 0.
\end{equation}
Next fix a codebook $\mathcal{C}$ and let $P_1$ denote the distribution resulting from conditioning the corresponding $\bar{P}_{Z^n}$ on the event $\bar{\mathcal{E}}_4$, and $P_2$ that on $\mathcal{E}_4$. Then $\bar{P}_{Z^n} = (1-P(\mathcal{E}_4|\mathcal{C}))P_1 + P(\mathcal{E}_4|\mathcal{C}) P_2$ and $\widehat{P}_{Z^n} = (1-P(\mathcal{E}_4|\mathcal{C}))P_1 + P(\mathcal{E}_4|\mathcal{C}) Q_0^{\times n}$. We have
\begin{IEEEeqnarray}{rCl}
\| \widehat{P}_{Z^n}-Q_0^{\times n}\|_{\mathrm{TV}} & = & (1-P(\mathcal{E}_4|\mathcal{C})) \| P_1 - Q_0^{\times n}\|_\mathrm{TV}\\
& \le & \| P_1 - Q_0^{\times n}\|_\mathrm{TV}.\label{eq:Gausscovert2}
\end{IEEEeqnarray}
On the other hand
\begin{IEEEeqnarray}{rCl}
\| \bar{P}_{Z^n} - Q_0^{\times n}\|_\mathrm{TV} & \ge & (1-P(\mathcal{E}_4|\mathcal{C})) \| P_1 - Q_0^{\times n}\|_\mathrm{TV} - P(\mathcal{E}_4|\mathcal{C}) \| P_2 - Q_0^{\times n}\|_\mathrm{TV}\\*
& \ge &  \| P_1 - Q_0^{\times n}\|_\mathrm{TV} -2 P(\mathcal{E}_4|\mathcal{C}).\label{eq:Gausscovert3}
\end{IEEEeqnarray}
Combining \eqref{eq:Gausscovert2} and \eqref{eq:Gausscovert3} we obtain
\begin{equation}
\E_{\calC}\| \widehat{P}_{Z^n}-Q_0^{\times n}\|_{\mathrm{TV}} \le \E_{\calC}\|\bar{P}_{Z^n}-Q_0^{\times n}\|_{\mathrm{TV}}+2P(\mathcal{E}_4).\label{eq:compareTV}
\end{equation}
By \eqref{eq:Gausscovert1}, \eqref{eq:compareTV}, and the fact that $P(\mathcal{E}_4)$ tends to zero as $n\to\infty$, we know that the left-hand side of \eqref{eq:compareTV} tends to zero as $n\to\infty$. Because $\widehat{P}_{Z^n}$ is absolutely continuous with respect to $Q_0^{\times n}$, this further implies that (see \cite[Section VII-A]{7707384})
\begin{equation}
D\left(\left.\widehat{P}_{Z^n}\right\|Q_0^{\times n}\right) \overset{n\rightarrow \infty} \longrightarrow 0.
\end{equation}

%

We next describe the decoder and analyze its probability of making an error. To this end, we first quantize the random variables $S$, $U$, and $Y$. A partition $\mathcal{P}$ of $\mathcal{U}$ is a finite collection of disjoint sets $P_i$ such that $\cup_i P_i=\mathcal{U}$. The quantization of $U$ by $\mathcal{P}$ is denoted as $[U]_{\mathcal{P}}$ and defined by 
\begin{align}
P\left( [U]_{\mathcal{P}}=\begin{cases}
\sup P_i & \mbox{if } \sup P_i < \infty\\
\inf P_i & \mbox{otherwise}
\end{cases}\right)=P(U\in P_i).
\end{align}
Similarly, $S$ (resp. $Y$) is quantized by partition $\tilde{\mathcal{P}}$ (resp. $\mathcal{P}'$) and its quantization is denoted by $[S]_{\tilde{\mathcal{P}}}$ (resp. $[Y]_{\mathcal{P}'}$). 
The decoder considers the above quantizations of the received sequence $y^n$ and every $u^n$ in the codebook, and performs typicality decoding as in the proof of Theorem~\ref{thm:achie}. The event $\mathcal{E}_4$ defined above, which has vanishing probability as $n\to\infty$, can be taken into account as an additional error event. Then the conditions \eqref{eqn:cover1} and \eqref{eqn:4rel} become
\begin{align}
R' & > I([U]_{\mathcal{P}};[S]_{\tilde{\mathcal{P}}})\\
R+R'&<I([U]_{\mathcal{P}};[Y]_{\mathcal{P}'}). 
\end{align} 
As we refine the partitions $\mathcal{P}$, $\tilde{\mathcal{P}}$, and $\mathcal{P}'$,  $I([U]_{\mathcal{P}};[S]_{\tilde{\mathcal{P}}})$ approaches $I(U;S)$ and $I([U]_{\mathcal{P}};[Y]_{\mathcal{P}'})$ approaches $I(U;Y)$ according to \cite[Section 8.6]{CoverThomas:06}.

We thus conclude that our coding scheme will succeed if \eqref{eqn:main_rate} and \eqref{eqn:main_key} hold for the chosen PDFs. Computing these expressions explicitly completes the achievability proof of Theorem~\ref{thm:Gaussian}.
\end{proof} 
\begin{proof}[Converse proof of Theorem \ref{thm:Gaussian}]  First, by examining the proof of Theorem~\ref{thm:converse}, we see that it also applies to the Gaussian channel. Fix the conditional distribution $P_{U|S}$ and the mapping $x(u,s)$ that achieve the maximum in \eqref{eqn:main_cv}. Recall that they satisfy $P_{Z}=Q_0$ and $\E[X^2]\leq P$.
Let $\tilde{P}:=\var(X)$ and $\Lambda:=\E[XS]$. It follows that
\begin{align}
I(U;Y)-I(U;S) &\leq I(U;Y,S)-I(U;S)\\
&=I(U;Y|S)\\
&\leq I(X,U;Y|S)\\
&\overset{(a)}=I(X;Y|S)\\
&=h(X+N_Y|S)-h(N_Y)  \\
&\overset{(b)}\leq \frac{1}{2}\log \left(1+\tilde{P}-\frac{\Lambda^2}{T}\right), \label{eqn:rate}
\end{align}
where $(a)$ is due to the Markov chain $U-(X,S)-Y$ and $(b)$ is from~\cite[Problem 2.7]{ElGamalKim:11}. Recall the condition $P_Z=Q_0$, which implies
\begin{align}
T+\sigma^2 & = \var(X+S+N_Z)\\
& = \tilde{P}+T+2\Lambda+\sigma^2,
\end{align}
therefore we must have $\Lambda=-\frac{\tilde{P}}{2}$. 
Hence \eqref{eqn:rate} implies
\begin{align}
I(U;Y)-I(U;S) &\leq \frac{1}{2}\log \left(1+\tilde{P}-\frac{\tilde{P}^2}{4T}\right).
\end{align}
Note that $\tilde{P}\leq P$ and 
\begin{align}
\argmax_{0\leq \tilde{P} \leq P } \left(\tilde{P}-\frac{\tilde{P}^2}{4T}\right) =\min\{P, 2T\}. 
\end{align}
Thus,  we have
\begin{align}
C\leq \frac{1}{2}\log \left(1+\min\{P, 2T\}-\frac{(\min\{P, 2T\})^2}{4T}\right),
\end{align}
which concludes the proof. 
\end{proof}

\begin{proof}[Proof of Theorem \ref{thm:Gaussian_causal}] 
We can adapt Theorem~\ref{thm:achie_causal} to the Gaussian case with a power constraint through a quantization argument that is similar to the one in the achievability proof of Theorem~\ref{thm:Gaussian}. 
By letting $V\sim \calN(0, P^*)$ and $X=V-\gamma^* S$ in Theorem~\ref{thm:achie_causal}, Theorem \ref{thm:Gaussian_causal} is proved. 
\end{proof}


\section{Concluding Remarks}\label{sec:conclusion}

We have shown that causal and noncausal CSI at the transmitter can sometimes help it to communicate covertly at a postive rate over channels which, without CSI, obey the ``square-root law'' for covert communications. Computable single-letter formulas for the maximum achievable covert-communication rate  (assuming that a sufficiently long key is available) have been derived. This work, from a different perspective to that of recent works \cite{chebakshichanjaggi14,7084182,SobersTowley:arxiv16}, shows that channel statistics unknown to the warden can help the communicating parties to communicate covertly.

There are many channels over which, even with the help of CSI, covert communication cannot have a positive rate   (Remark~\ref{rmk:nogain} contains simple examples). For some of these channels, CSI may help to improve the scaling constant of the maximum amount of information that can be covertly communicated with respect to the square root of the total number of channel uses. We have not investigated this possibility in the current paper.

So far, we have not been able to prove upper bounds on the minimum secret-key length required to achieve the covert capacity that match the lower bounds \eqref{eqn:main_key_causal} and \eqref{eqn:main_key}, except when the warden has a weaker channel than the intended receiver, in which case this length is zero. This key-length problem may be related to the secrecy capacity of the wiretap channel with causal or noncausal CSI at the transmitter \cite{6135500,GoldfeldCuffPermuter:arxiv16}, which, to the best of our knowledge, is not yet completely solved.

\end{document}